\documentclass[article,a4paper,11pt]{article}
\usepackage[english]{babel}
\usepackage[T1]{fontenc}
\usepackage[bitstream-charter]{mathdesign}

\usepackage{multicol}
\usepackage{amsmath}
\usepackage{amssymb}
\usepackage{amsthm}
\usepackage{graphicx}
\usepackage{xcolor}
\usepackage{authblk}
\usepackage{marvosym}
\usepackage{cite}
\usepackage{kotex}
\usepackage{vmargin}
\setmarginsrb{2cm}{2cm}{2cm}{2cm}{0mm}{0mm}{0mm}{10mm}
\usepackage{gensymb}
\usepackage[colorlinks=true, allcolors=blue, breaklinks=true]{hyperref}
\theoremstyle{definition}
\newtheorem{definition}{Definition}
\newtheorem{theorem}{Theorem}

\newtheorem{proposition}{Proposition}
\usepackage{url}
\usepackage{tikz}
\usetikzlibrary{arrows.meta}
\newcommand{\sech}{\mathrm{sech} \,}

\title{\vspace*{-1.5cm} \bfseries Fourier spectrum and related characteristics of the fundamental bright soliton solution }
\author{\normalsize Natanael Karjanto\thanks{\Letter: \texttt{natanael@skku.edu} \href{https://orcid.org/0000-0002-6859-447X}{\includegraphics[scale=0.08]{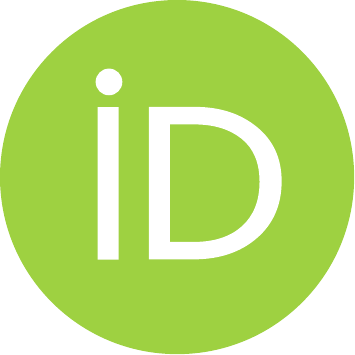}}}}
\affil{Department of Mathematics, University College, Natural Science Campus\\ Sungkyunkwan University, Suwon~16419, Republic of Korea}

\date{\vspace*{-0.5cm} \scriptsize Updated \today}

\begin{document}
\maketitle

\vspace*{-0.5cm}
\begin{abstract}
\noindent
We derive exact analytical expressions for the spatial Fourier spectrum of the fundamental bright soliton solution for the $(1 + 1)$-dimensional nonlinear Schr\"odinger equation. Similar to a Gaussian profile, the Fourier transform for the hyperbolic secant shape is also shape-preserving. We further confirm that the fundamental soliton indeed satisfies essential characteristics such as Parseval's relation and the stretch-bandwidth reciprocity relationship. The fundamental bright solitons find rich applications in nonlinear fiber optics and optical telecommunication systems.\\

\end{abstract}

\section{Introduction}

The nonlinear Schr\"odinger (NLS) equation is a nonlinear evolution equation for slowly varying wave packet envelopes in dispersive media. It belongs to a category of completely integrable systems or exactly solvable models of a nonlinear partial differential equation (PDE), with infinitely many conserved quantities and explicit analytical solutions. The NLS equation finds rich applications in mathematical physics, including surface gravity waves, nonlinear optics, superconductivity, and Bose-Einstein condensates (BEC)~\cite{malomed2005nonlinear,debnath2012nonlinear,karjanto2020nonlinear}. 

The focus of this article is the one-dimensional focusing type of the NLS equation. The former is sometimes also written as $(1 + 1)$-D, where usually, there is only one dimension each in the spatial and temporal variables. Higher-dimensional models usually extend the dimension of the dispersive term. In the canonical form, the focusing $(1 + 1)$-D NLS equation can be expressed as follows~\cite{ablowitz1981solitons,sulem1999nonlinear,fibich2015nonlinear}:
\begin{equation}
i q_t + q_{xx} + 2|q|^2 q = 0, \qquad (x,t) \in \mathbb{R}^2, \qquad q(x,t) \in \mathbb{C}.		\label{NLSmodel}
\end{equation}

In surface gravity waves, the independent variables $(x,t)$ represent spatial and temporal quantities, respectively, whereas in nonlinear optics, $t$ denotes the  transversal pulse propagation in space and $x$ designates the time variable. The focusing type is indicated by the positive product of the dispersive and nonlinear coefficients, which are also known in optical pulse propagation as group-velocity dispersion and self-phase modulation, respectively. In this article, we adopt the physical interpretation from hydrodynamics for the independent variables.

The dependent variable $q(x,t)$ is a complex-valued amplitude. It describes slowly-varying envelope dynamics of the corresponding weakly nonlinear quasi-monochromatic wave packet profile. Denoted by $\eta(\bar{x},\bar{t})$, the usual relationship with the $q$ is given by
\begin{equation*}
\eta(\bar{x},\bar{t}) = \text{Re} \left\{q(x,t) e^{i(k \bar{x} - \omega \bar{t})} \right\},
\end{equation*}
where $x = \varepsilon(\bar{x} - c_g \bar{t})$ and $t = \varepsilon^2 \bar{t}$, with $\varepsilon \ll 1$. Here, $c_g$ denotes the group velocity, and the wavenumber $k$ and wave frequency $\omega$ are related by the linear dispersion relationship for the corresponding medium.

The purpose of this article is to provide a step-by-step explanation for finding the physical spectrum, i.e., the spatial Fourier spectrum, of the bright soliton solution of the NLS equation. The term ``spectrum'' should not be confused with the one used in the well-known inverse scattering transform (IST) technique~\cite{ablowitz1981solitons,ablowitz1991solitons,osborne2010nonlinear}. For working definitions, we adopt the ones from~\cite{pelinovsky2005spectral,karjanto2021spatial}. 

The article is organized as follows. Section~\ref{spafoutra} outlines the proof of the spatial Fourier transform for the fundamental bright soliton. Section~\ref{property} discusses some essential characteristics of the soliton, including Parseval's theorem and the stretch-bandwidth reciprocity relationship. Section~\ref{conclude} concludes our discussion.

\section{Spatial Fourier spectrum}		\label{spafoutra}

For the NLS equation~\eqref{NLSmodel}, the simplest form of the fundamental soliton solution is given by~\cite{agrawal2019nonlinear}
\begin{equation*}
q(x,t) = e^{it} \, \sech x = \frac{e^{it}}{\cosh x}.	\label{bright}
\end{equation*}
Also called a bright soliton, this exact analytic expression was discovered one-half century ago using the IST~\cite{zakharov1972exact}. One may also apply the Darboux transformation using the seed function $q = 0$ to obtain this solution~\cite{matveev1991darboux,trisetyarso2009application}. Without employing the IST, the bright soliton can also be obtained by solving the NLS equation directly by assuming the existence of a shape-preserving solution in the form of the phase and time-independent amplitude~\cite{agrawal2019nonlinear}. The term \emph{bright} soliton is often used in the nonlinear optics literature to distinguish and contrast it with \emph{dark} soliton. The former exists in the anomalous dispersion regime, modeled by the focusing NLS equation~\eqref{NLSmodel}, whereas the latter occurs under the normal dispersion regime, which is governed by the defocusing NLS equation~\cite{banerjee2004nonlinear,li2017nonlinear}.

This fundamental bright soliton solution occurs as limiting cases of a family of stationary periodic wave solutions, another family of periodic solutions, in which both involve Jacobi elliptic functions, and the Kuznetsov-Ma breather family~\cite{akhmediev1997solitons,kuznetsov1977solitons,ma1979perturbed,karjanto2021peregrine}. Although the possibility for the formation of the bright soliton was suggested as early as 1973, it was not until 1980 that its appearance was observed experimentally in optical fibers~\cite{hasegawa1973transmission1,hasegawa1973transmission2,mollenauer1980experimental,taylor1992optical}. Indeed, this fundamental soliton has a remarkable feature that allures it for practical application, i.e., if a hyperbolic secant pulse is initiated inside an ideal lossless fiber, it would propagate without altering its shape for arbitrarily long-distance~\cite{nrc1997nonlinear,kath1998making,christiansen2000nonlinear,millot2005solitons}. In BEC, the formation of matter-wave bright solitons was observed in 2002~\cite{khaykovic2002formation,strecker2002formation}.

For $a$, $k$, and $\lambda \in \mathbb{R} \backslash \{0\}$, the following transformations of the fundamental bright soliton~\eqref{bright} will also satisfy the NLS equation~\eqref{NLSmodel}~\cite{dysthe1999note}:
\begin{align*}
q_1(x,t) &= a q(ax, at^2), \\
q_2(x,t) &= q(x - 2kt, t) e^{i \left(kx - k^2 t + \lambda \right)}, \\
q_3(x,t) &= a q(a(x - 2kt), at^2) e^{i \left(ak x - a k^2 t^2 + \lambda \right)}.
\end{align*}
These three arbitrary parameters characterize and are related to the amplitude, frequency, and phase of the soliton, respectively. The fourth parameter is absent and might always be re-introduced if one wishes to indicate the position of the soliton peak. However, this is inessential as we can always shift it to $x = 0$ when $t = 0$~\cite{agrawal2019nonlinear}.

We have the following theorem.
\begin{theorem}
The spatial Fourier spectrum of the bright NLS soliton in its simplest form~\eqref{bright} is given by
\begin{equation*}
\widehat{q}(k,t) = \pi \, \sech \left(\frac{\pi}{2} k \right) e^{it}.
\end{equation*}
\end{theorem}

\begin{proof}
Observe that if we let 
\begin{equation*}
f(x) = \sech x = \frac{1}{\cosh x} = \frac{2}{e^x + e^{-x}} = \frac{2 e^{-x}}{1 + e^{-2x}},
\end{equation*}
then $f$ is an even function and 
\begin{equation*}
\left|f(x) \right| = \frac{2 e^{-|x|}}{1 + e^{-2|x|}} \leq 2 e^{-|x|} \in L_1(\mathbb{R}). 
\end{equation*}

We are interested in finding the spatial Fourier spectrum of the bright soliton $q(x,t)$.
\begin{equation*}
\widehat{q}(k,t) = \int_{-\infty}^{\infty} \sech x e^{-i(k x - t)} \, dx
= e^{it} \left(\int_{-\infty}^{\infty} \frac{\cos kx}{\cosh x} \, dx - i \int_{-\infty}^{\infty} \frac{\sin kx}{\cosh x} \, dx \right).
\end{equation*}
Since the second term inside the bracket of the last expression vanishes, we only need to evaluate
\begin{equation*}
\int_{-\infty}^{\infty} \frac{\cos kx}{\cosh x} \, dx = \lim\limits_{R \to \infty} \int_{-R}^{R} \frac{2\cos kx}{e^{x} + e^{-x}} \, dx, \qquad R > 0.
\end{equation*}
Consider the complex-valued function
\begin{equation}
\phi(z) = \frac{2\cos kz}{e^{z} + e^{-z}}, \qquad z \in \mathbb{C},		\label{phi}
\end{equation}
and the rectangular contour $C_R$ in the complex plane with corners at $\pm R$ and $\pm R + i\pi$, with $R > 0$, as shown in Figure~\ref{contour}.
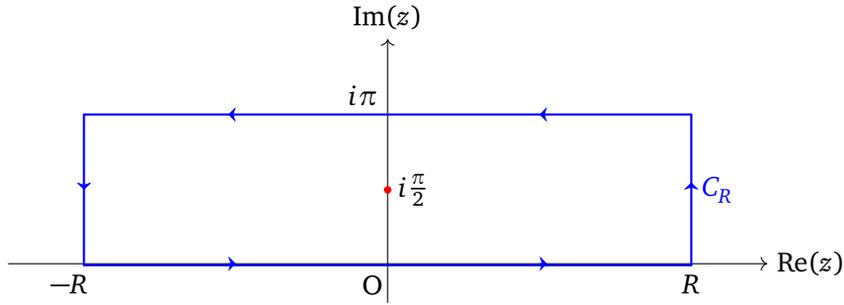
\begin{figure}[htbp]
\begin{center}
\begin{tikzpicture}
\draw[->] (-5,-0.98)--(5,-0.98) node[right]{Re($z$)};
\draw[->] (0,-1.5)--(0,2) node[above]{Im($z$)};
\node[draw,thick,rectangle,minimum width=8cm,minimum height=2cm,color=blue](r){};
\draw[-{Stealth[length=1mm,width=2mm]},color=blue] (4,0)--(4,0.1);
\draw[-{Stealth[length=1mm,width=2mm]},color=blue] (-4,0.1)--(-4,0);
\draw[-{Stealth[length=1mm,width=2mm]},color=blue] (2.1,1)--(2,1);
\draw[-{Stealth[length=1mm,width=2mm]},color=blue] (-2,1)--(-2.1,1);
\draw[-{Stealth[length=1mm,width=2mm]},color=blue] (-2.1,-0.98)--(-2,-0.98);
\draw[-{Stealth[length=1mm,width=2mm]},color=blue] (2,-0.98)--(2.1,-0.98);
\node[below] at (-0.2,-1) {O};
\node[below] at (4,-1) {$R$};
\node[below] at (-4.2,-1) {$-R$};
\node[left] at (0,1.25) {$i \pi$};
\node[circle, fill, inner sep=1pt, color=red] at (0,0) {};
\node[right] at (0,0) {$i\frac{\pi}{2}$};
\node[right,color=blue] at (4,0) {$C_{R}$};
\end{tikzpicture}
\end{center}
\caption{A rectangular contour $C_{R}$ with corners $\pm R$ and $\pm R + i\pi$ determined by a closed, simple, piecewise smooth curve traversed in positive orientation. The function $\phi$ is holomorphic in an open connected domain containing the interior of $C_R$ and its closure, except at the pole $z = i \frac{\pi}{2}$.}		\label{contour}
\end{figure}

The function $\phi(z)$~\eqref{phi} has only one pole inside the region bounded by the rectangular contour $C_R$ at $z = i \frac{\pi}{2}$. Hence, according to the consequence of the Residue Theorem~\cite{howie2003complex,brown2014complex,ablowitz2021introduction}, we can calculate that the contour integral over $C_R$ is given by
\begin{equation}
\int_{C_R} \phi(z) \, dz = 2\pi i \; \text{res}\left(\phi, i \frac{\pi}{2} \right) 
= 2 \pi i \; \frac{2 \cos \left(k i \frac{\pi}{2}\right)}{e^{i \frac{\pi}{2}} - e^{-i \frac{\pi}{2}}}
= 2 \pi i \; \frac{2 \cosh\left(\frac{\pi}{2} k \right)}{2 i \sin \frac{\pi}{2}}
= 2 \pi \cosh\left(\frac{\pi}{2} k \right). 		\label{contourintegral1}
\end{equation}
On the other hand, this very same contour integral can also be expressed as separate four distinct line integrals along each side of the rectangle. Hence, we have
\begin{align}
\int_{C_R} \phi(z) \, dz 
&= \int_{-R}^{R} \frac{2\cos kx}{e^{x} + e^{-x}} \, dx 
+ \int_{0}^{\pi} \frac{2\cos k(R + iy)}{e^{R + iy} + e^{-R - i y}} \, i \, dy
- \int_{-R}^{R} \frac{2\cos k(x + \pi \, i)}{e^{x + \pi \, i} + e^{-x - \pi \, i}} \, dx \nonumber \\
& \quad - \int_{0}^{\pi} \frac{2\cos k(-R + iy)}{e^{-R + iy} + e^{R - i y}} \, i \, dy.		\label{contourintegral2}
\end{align}
Now observe that the integrand of the second integral can be expressed as follows for sufficiently large $R$:
\begin{align*}
\left|\frac{2\cos k(R + iy)}{e^{R + iy} + e^{-R - i y}} \right| 
&= \left|\frac{e^{i k R} \, e^{-ky} + e^{-i k R} \, e^{ky}}{e^{R + iy} + e^{-R - i y}} \right|
= \left|\frac{e^{(1 + ik)R} \, e^{(-k + i)y} + e^{(1 - ik)R} \, e^{(k + i)y}}{e^{2(R + iy)} + 1} \right|  \\
&\leq \left|\frac{e^{(1 + ik)R} \, e^{(-k + i)y}}{e^{2(R + iy)} + 1} \right| + \left| \frac{e^{(1 - ik)R} \, e^{(k + i)y}}{e^{2(R + iy)} + 1} \right|  
\leq \frac{e^{R - ky}}{e^{2R} - 1} + \frac{e^{R + ky}}{e^{2R} - 1}
= \frac{2 e^R \, \cosh ky}{e^{2R} - 1} \\
&\leq 4 e^{-R} \cosh ky.
\end{align*}
Thus, the absolute value of the second integral reads
\begin{align*}
\left| \int_{0}^{\pi} \frac{2\cos k(R + iy)}{e^{R + iy} + e^{-R - i y}} \, i \, dy \right|
\leq \int_{0}^{\pi} \left| \frac{2\cos k(R + iy)}{e^{R + iy} + e^{-R - i y}} \, i \right| \, dy
\leq \int_{0}^{\pi} 4 e^{-R} \cosh ky \, dy = 4 e^{-R} \sinh k\pi,
\end{align*}
which tends to $0$ as $R \to \infty$. Similarly, for sufficiently large $R$, the function of the fourth integral satisfies the following relationship:
\begin{align*}
\left|\frac{2\cos k(-R + iy)}{e^{-R + iy} + e^{R - i y}} \right| 
&= \left|\frac{e^{-i k R} \, e^{-ky} + e^{i k R} \, e^{ky}}{e^{-R + iy} + e^{R - i y}} \right|
= \left|\frac{e^{-(1 + ik)R} \, e^{(-k + i)y} + e^{(-1 + ik)R} \, e^{(k + i)y}}{1 + e^{-2(R - iy)}} \right|  \\
& \leq \frac{\left|e^{-(1 + ik)R} \, e^{(-k + i)y} \right| + \left|+ e^{(-1 + ik)R} \, e^{(k + i)y} \right|}{\left|1 + e^{-2(R - iy)} \right|}
\leq \frac{e^{-R - ky} + e^{-R + ky}}{1 - e^{-2R}} = \frac{2e^{-R} \cosh ky}{1 - e^{-2R}} \\
& \leq 2e^{-R} \cosh ky.
\end{align*} 
So, the absolute value of the fourth integral becomes
\begin{align*}
\left| \int_{0}^{\pi} \frac{2\cos k(-R + iy)}{e^{-R + iy} + e^{R - i y}} \, i \, dy \right|
\leq \int_{0}^{\pi} \left| \frac{2\cos k(-R + iy)}{e^{-R + iy} + e^{R - i y}} \, i \right| \, dy
\leq \int_{0}^{\pi} 2 e^{-R} \cosh ky \, dy = 2 e^{-R} \sinh k\pi,
\end{align*}
which again tends to $0$ as $R \to \infty$. Considering the third integral, it follows that
\begin{align*}
-\int_{-R}^{R} \frac{2\cos k(x + \pi \, i)}{e^{x + \pi \, i} + e^{-x - \pi \, i}} \, dx
&= \int_{-R}^{R} \frac{e^{ikx} e^{-\pi k} + e^{-ikx} e^{\pi k}}{e^{x} + e^{-x}} \, dx 
= \int_{-R}^{R} \frac{\cos kx \, \cosh \pi k}{\cosh x} \, dx - i \int_{-R}^{R} \frac{\sin kx \, \sinh \pi k}{\cosh x} \, dx \\
&=  \cosh \pi k \int_{-R}^{R} \frac{\cos kx }{\cosh x} \, dx,
\end{align*}
where the imaginary component of the integral vanishes since $\sin(kx)$ is an odd function. Hence, letting $R \to \infty$, we observe from~\eqref{contourintegral1} and~\eqref{contourintegral2} that
\begin{equation*}
\lim\limits_{R \to \infty} \left(1 + \cosh \pi k \right) \int_{-R}^{R} \frac{\cos kx}{\cosh x} \, dx 
= \left(1 + \cosh \pi k \right) \int_{-\infty}^{\infty} \frac{\cos kx}{\cosh x} \, dx 
= 2 \pi \cosh \left(\frac{\pi}{2} k \right).
\end{equation*}
Therefore,
\begin{align*}
\int_{-\infty}^{\infty} \frac{\cos kx}{\cosh x} \, dx 
&= \frac{2 \pi \cosh \left(\frac{\pi}{2} k \right)}{1 + \cosh \pi k}
 = \frac{2 \pi \cosh \left(\frac{\pi}{2} k \right)}{2 \cosh^2 \left(\frac{\pi}{2} k \right)}
 = \frac{\pi}{\cosh \left(\frac{\pi}{2} k \right)}
 = \pi \, \sech \left(\frac{\pi}{2} k\right),
\end{align*}
and thus, the spatial Fourier transform for the bright soliton reads
\begin{equation*}
\widehat{q}(k,t) = \pi \, \sech \left(\frac{\pi}{2} k\right) e^{it}.
\end{equation*}
This completes the proof.
\end{proof}
\begin{figure}[h]
\begin{center}
\includegraphics[width=0.45\textwidth]{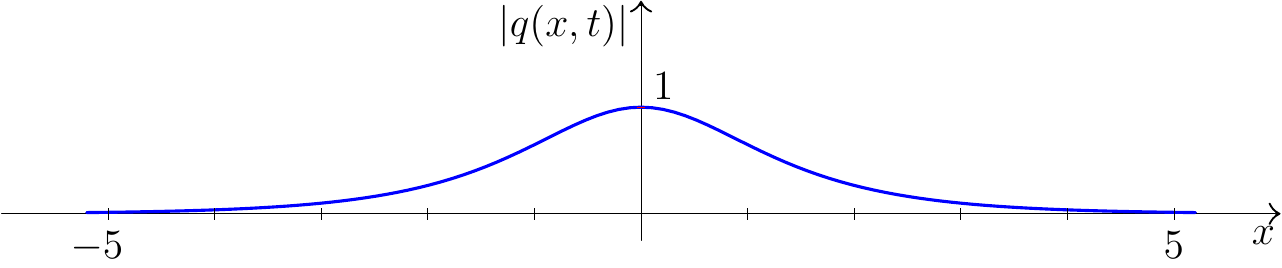}	\hspace*{1cm}
\includegraphics[width=0.45\textwidth]{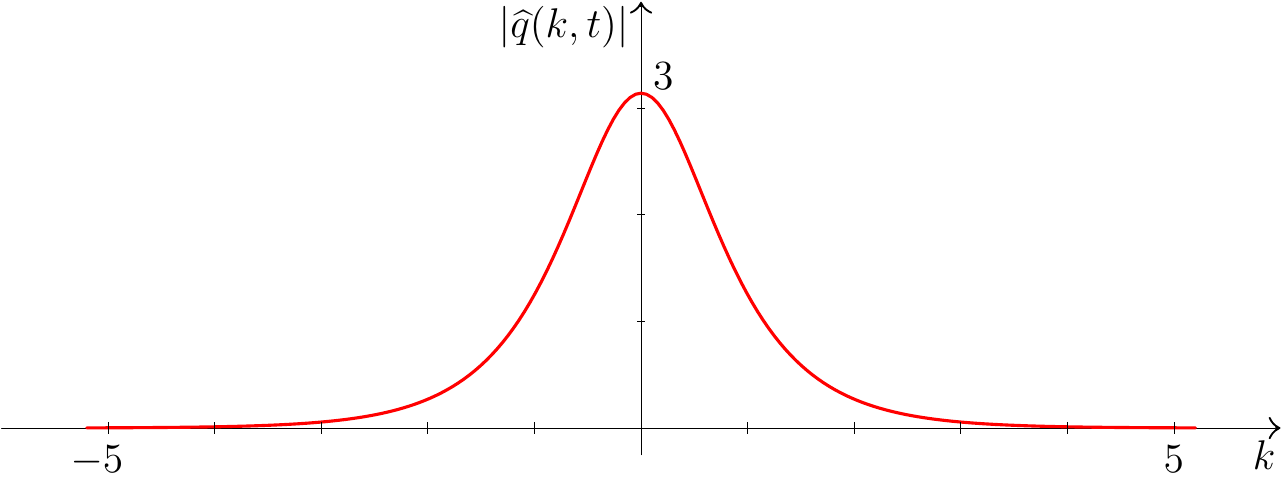}
\end{center}
\caption{The moduli of the NLS bright soliton (left panel, blue curve) and its corresponding spectrum (right panel, red curve). Although both profiles are secant hyperbolic, the former is wider and shorter, whereas the latter is narrower and taller. Both plots are depicted on the same axis scale.}		\label{profile}
\end{figure}
A similar approach in calculating the bright soliton spectrum yields an expression of convergent alternating series form, where the integral in the associated Fourier transform is expressed as Gauss' hypergeometric function~\cite{karjanto2006thesis}:
\begin{equation*}
\widehat{q}(k,t) = 4 \, e^{it} \sum_{n = 0}^{\infty} (-1)^n \frac{2n + 1}{k^2 + (2n + 1)^2}.
\end{equation*}

\section{Soliton characteristics}	\label{property}
In this section, we examine some characteristics of the fundamental bright soliton and its corresponding spatial Fourier spectrum. While the literature usually considers the temporal-frequency domains relationship, i.e., pulse or signal (together with its envelope) in the time domain and its temporal Fourier spectrum in the frequency domain, we focus on the relationship of the spatial-wavenumber domains. Hence, some terminologies in the definition are adjusted to the related domains accordingly. We commence with the following definitions and propositions.

\begin{definition}[Soliton power and energy spectral density]
The squared-absolute value of the fundamental bright soliton $\left|q(x,t)\right|^2$ is called the \emph{(envelope) soliton power}, whereas the squared-modulus of its spatial Fourier spectrum $\left|\widehat{q}(k,t)\right|^2$ is called the \emph{energy spectral density}.
\end{definition}

Figure~\ref{profile} displays the envelope power of the bright soliton (left panel) and its associated energy spectral density (right panel) for a fixed value of $t \in \mathbb{R}$. We observe that the Fourier transform preserves the hyperbolic secant profile albeit with different height and width between the spatial and wavenumber domains.

\begin{proposition}		\label{parsev}
The fundamental bright soliton satisfies Parseval's Theorem in the spatial-wavenumber domains~\cite{oppenheim1997signals}:
\begin{equation*}
\int_{-\infty}^{\infty} \left|q(x,t) \right|^2 \, dx = \frac{1}{2\pi} \int_{-\infty}^{\infty} \left|\widehat{q}(k,t) \right|^2 \, dk.
\end{equation*}	
\end{proposition}
\begin{proof}
The left-hand side reads
\begin{align*}
\int_{-\infty}^{\infty} \left|q(x,t) \right|^2 \, dx 
&= \int_{-\infty}^{\infty} \sech^2 x \, dx 
 = \lim\limits_{c \to \infty} \int_{-c}^{c} \sech^2 x \, dx 
 = \lim\limits_{c \to \infty}  \tanh x \, \Big|_{-c}^c \\ 
&= \lim\limits_{c \to \infty}  \left[\tanh c - \tanh (-c) \right] 
 = \left[1 - (-1) \right] = 2.
\end{align*}
The right-hand side calculates
\begin{align*}
\frac{1}{2 \pi}\int_{-\infty}^{\infty} \left|\widehat{q}(k,t) \right|^2 \, dk 
&= \frac{1}{2 \pi} \int_{-\infty}^{\infty} \pi^2 \sech^2 \left(\frac{\pi}{2} k \right) \, dk 
 = \frac{\pi}{2} \lim\limits_{d \to \infty} \int_{-d}^{d} \sech^2 \left(\frac{\pi}{2} k \right) \, dk 
 = \lim\limits_{d \to \infty}  \tanh \left(\frac{\pi}{2} k \right) \Bigg|_{-d}^d \\
&= \lim\limits_{d \to \infty}  \left[\tanh \left(\frac{\pi}{2} d \right) - \tanh \left(-\frac{\pi}{2} d \right) \right] 
 = \left[1 - (-1) \right] = 2.
\end{align*}
Since we verify that both sides equal to two, we complete the proof.
\end{proof}
\begin{definition}[Mean/centroid]	\label{centro}
Let $q(x,t)$ and $\widehat{q}(k,t)$ be a complex-valued soliton amplitude and its associated spatial Fourier spectrum, respectively.	
Then, the \emph{soliton mean} or \emph{centroid} is defined by
\begin{equation*}
\bar{x} = \frac{\displaystyle \int_{-\infty}^{\infty} x \left|q(x,t)\right|^2\, dx}{\displaystyle \int_{-\infty}^{\infty} \left|q(x,t)\right|^2\, dx},
\end{equation*}
whereas the (spatial) \emph{spectrum centroid} or \emph{mean} is defined as
\begin{equation*}
\bar{k} = \frac{\displaystyle \int_{-\infty}^{\infty} k \left|\widehat{q}(k,t)\right|^2 \, dk}{\displaystyle \int_{-\infty}^{\infty} \left|\widehat{q}(k,t)\right|^2 \, dk}.
\end{equation*}
\end{definition}
\begin{definition}[Power-root-mean-square soliton width and spectral bandwidth]
Let $q(x,t)$ and $\widehat{q}(k,t)$ be a complex-valued soliton amplitude and its associated spatial Fourier spectrum, respectively.  
The \emph{power-root-mean-square soliton width} $\sigma_x \geq 0$, also called the \emph{soliton standard deviation} or \emph{soliton radius of gyration}, is defined as
\begin{equation*}
\sigma_{x}^2 = \frac{\displaystyle \int_{-\infty}^{\infty} (x - \bar{x})^2 \left|q(x,t)\right|^2 \, dx}{\displaystyle \int_{-\infty}^{\infty} \left|q(x,t)\right|^2 \, dx},
\end{equation*}
where as the (spatial) \emph{power-root-mean-square spectral bandwidth} $\sigma_{k} \geq 0$ is defined by
\begin{equation*}
\sigma_{k}^2 = \frac{\displaystyle \int_{-\infty}^{\infty} (k - \bar{k})^2 \left|\widehat{q}(k,t)\right|^2 \, dk}{\displaystyle \int_{-\infty}^{\infty} \left|\widehat{q}(k,t)\right|^2 \, dk}.
\end{equation*}
In both instances, $\bar{x}$ and $\bar{k}$ denote the soliton and spectrum centroids mentioned earlier in Definition~\ref{centro}, respectively.
\end{definition}
\begin{proposition}
The product of the power-root-mean-square soliton width with its associated spectral bandwidth satisfies the (spatial) \emph{stretch-bandwidth reciprocity relationship}, also called the \emph{space-wavenumber uncertainty relationship}~\cite{hirlimann2005pulsed,saleh2007fundamentals}:
\begin{equation*}
\sigma_x \, \sigma_{k} \geq \frac{1}{2}.
\end{equation*}
\end{proposition}
\begin{proof}
Since $x \left|q(x,t)\right|^2 = x \, \sech^2 x$ and $k \left|\widehat{q}(k,t)\right|^2 = \pi^2 \, k \, \sech^2 \left(\frac{\pi}{2} k \right)$ are both odd functions with respect to the length $x$ and wavenumber $k$, respectively, then the numerators of both the soliton's and spectrum's means vanish, i.e., $\bar{x} = 0 = \bar{k}$. From Proposition~\ref{parsev}, we calculated the denominators of $\sigma_{x}^2$ and $\sigma_{k}^2$ as $2$ and $4\pi$, respectively. It then remains to evaluate their numerators. The former reads
\begin{align*}
\int_{-\infty}^{\infty} (x - \bar{x})^2 \left|q(x,t)\right|^2 \, dx
&= \int_{-\infty}^{\infty} x^2 \, \sech^2 x \, dx
 = \lim\limits_{c \to \infty} \int_{-c}^{c} x^2 \, \sech^2 x \, dx \\
&= \lim\limits_{c \to \infty}  \text{Li}_2(-e^{-2x}) + x^2 \left(\tanh x - 1\right) - 2x \ln (1 + e^{-2x}) \, \Big|_{-c}^c 
 = \frac{\pi^2}{6}.
\end{align*}
Thus, the soliton variance $\sigma_{x}^2$ and its power-root-mean-square width $\sigma_{x}$ are given by
\begin{equation*}
\sigma_{x}^2 = \frac{\pi^2}{12} \qquad \text{and} \qquad \sigma_{x} = \frac{\pi}{2 \sqrt{3}}.
\end{equation*}
Similarly, we evaluate
\begin{align*}
\int_{-\infty}^{\infty} (k - \bar{k})^2 \left|\widehat{q}(k,t)\right|^2 \, dk
&= \int_{-\infty}^{\infty} \pi^2 \, k^2 \, \sech^2 \left(\frac{\pi}{2} k \right) \, dk
= \pi^2 \lim\limits_{d \to \infty} \int_{-d}^{d} k^2 \, \sech^2 \left(\frac{\pi}{2} k \right) \, dk \\
&= \frac{1}{\pi} \lim\limits_{d \to \infty}  8 \text{Li}_2(-e^{-\pi k}) + \pi^2 k^2 \left[\tanh \left(\frac{\pi}{2} k \right) - 2 \right] - 8 \pi k \ln (1 + e^{-\pi k}) \, \Big|_{-d}^d 
= \frac{4 \pi}{3}.
\end{align*}
It follows that the spectrum variance $\sigma_{k}^2$ and its power-root-mean-square spectral bandwidth $\sigma_{k}$ are given as follows:
\begin{equation*}
\sigma_{k}^2 = \frac{1}{3} \qquad \text{and} \qquad \sigma_{k} = \frac{1}{\sqrt{3}}.
\end{equation*}
In both instances, Li$_2(z)$ denotes the dilogarithm Spence's function, defined as
\begin{equation*}
\text{Li}_2(z) = - \int_{0}^{z} \frac{\ln (1 - u)}{u} \, du \qquad \text{or} \qquad
\text{Li}_2(z) =   \int_{1}^{z} \frac{\ln t}{1 - t} \, dt, \qquad z \in \mathbb{C}.
\end{equation*} 
Hence, the product of the two widths yields
\begin{equation*}
\sigma_{x} \, \sigma_{k} = \frac{\pi}{6} \geq \frac{1}{2},
\end{equation*}
which satisfies the stretch-bandwidth reciprocity relationship and completes the proof.
\end{proof}
It can also be verified that the equality in the stretch-bandwidth reciprocity relationship is satisfied for the Gaussian pulse profile. In other words, the Gaussian function possesses the minimum permissible value of the stretch-bandwidth product, i.e., the minimum uncertainty in the context of Heisenberg's uncertainty principle in quantum mechanics.

\section{Conclusion}	\label{conclude}

In this article, we have considered the spatial Fourier spectrum for the fundamental bright soliton solution of the NLS equation. Deriving the analytical expression of the spectrum requires performing integration in the complex plane. Interestingly, the bright soliton has a similar characteristic of hyperbolic secant profiles in both the spatial and wavenumber domains. The dynamics of the fundamental bright soliton is a topic of ongoing interest, thanks to its applications in optical telecommunication systems and nonlinear fiber optics.

\subsection*{Acknowledgment}
The author gratefully acknowledges E. (Brenny) van Groesen and Andonowati for the long-lasting guidance and concerted cultivation in thinking and growing up mathematically.

\subsection*{Conflict of Interest}
The author declares no conflict of interest.

{\small

\par}
\end{document}